\documentclass[runningheads,a4paper]{llncs}

\usepackage{amssymb}
\usepackage{graphicx}

\usepackage{url}
\urldef{\mailsa}\path|{alfred.hofmann, ursula.barth, ingrid.haas, frank.holzwarth,|
\urldef{\mailsb}\path|anna.kramer, leonie.kunz, christine.reiss, nicole.sator,|
\urldef{\mailsc}\path|erika.siebert-cole, peter.strasser, lncs}@springer.com|    

\spnewtheorem{observation}{Observation}{\upshape\bfseries}{\slshape}
\spnewtheorem{claimm}[proposition]{Claim}{\bfseries}{\itshape}

\newcommand{\ignore}[1]{}

\def \Lyap {\rm L}

\begin{document}

\mainmatter  


\title{No Ascending Auction can find Equilibrium for SubModular valuations}


%
%
\author{Oren Ben-Zwi\inst{1}
\and Ilan Newman\inst{2}}

%
%
\institute{Tel-Hai College, Department of Computer Science, Tel-Hai, Israel\\
\email {oren.benzwi@gmail.com}
\and
University of Haifa Department of Computer Science, Haifa, Israel\\
\email{ilan@cs.haifa.ac.il}}

\maketitle

\begin{abstract}
We show that no efficient ascending
auction can guarantee to find even a minimal envy-free price vector if all valuations are
submodular, assuming 
a basic complexity theory's assumption.
\end{abstract}

\section{Introduction}
Ascending auctions are arguably the most important auction format, and are used as a synonym to the word auction in many scenarios.
Some advantages are visibility of price formation and gradual disclosure of valuations. It is harder for a dishonest auctioneer to manipulate the bids in such auctions, and it is easier
to adjust such auctions to a setting with budget constraints. It is a popular auction format in practice, and its theoretical properties have been widely studied.

The main task of an auction designer is to come up with a mechanism that allocates the items to the bidders such that no bidder is dissatisfied, or have \emph{envy}.
An allocation and a price vector were every bidder is allocated its desired set of items is called \emph{envy free} allocation, and the price vector is an \emph{envy free} price vector.
An envy free price vector is \emph{minimal} if there is no (bit-wise) lower envy free price vector.

If for an envy free price vector all items can be allocated (or equally we sometimes say all unallocated items have price zero - as we assume monotonicity) the price vector is called \emph{Walrasian price vector}~\cite{Walras1874}. For a Walrasian price vector, the given allocation maximizes the \emph{social welfare} which is the sum of the bidders valuations of their allocated sets~\cite{Mas-CollelWhGr1995}.

We are interested in the dynamics by which a minimal envy free price vector is calculated given a \emph{combinatorial auction} where all the bidder valuations are \emph{submodular}.

For combinatorial auctions, different bidder valuation classes
were examined by Lehmann et. al, and a hierarchy of valuation
classes was presented~\cite{LehmannLeNi2006}. In this hierarchy unit demand
and additive valuation classes are contained in the
gross substitute class; and the gross substitute class is a
subclass of submodular valuations. All the inclusions are
strict. Kelso and Crawford proved that if all valuations are gross substitute then a Walrasian equilibrium is guaranteed~\cite{KelsoCr1982}; therefore the problems of finding a minimal envy free price vector, finding the maximum social welfare and finding a Walrasian price vector collide.

Demange~et~al.~(DGS)~\cite{DGS86} were the first to formally define an ascending auction that terminates in a Walrasian equilibrium. They analyzed
the case of unit-demand valuations. Andersson~et~al.~\cite{AnderssonAnTa2013} fully characterize the class of all ascending auctions that terminate in a Walrasian equilibrium,
for unit-demand bidders.
Moreover, we can think of the Hungarian
method for the assignment problem~\cite{Kuhn1955} as a polynomial
ascending auction for the maximum welfare problem
with unit demand bidders, where the final price (dual variables) is a minimal envy free price vector.

A combinatorial auction with additive valuation bidders can be reduced to multiple one item auctions, hence the English (ascending) auction which is equivalent to a second price auction for all the items one by one, finds a minimal envy free price in polynomial time.

For gross substitute valuations many ascending auctions were
suggested~\cite{KelsoCr1982,GulSt2000,Ausubel2005,Ben-ZwiLaNe2013},
and a thorough discussion of all ascending auction for gross substitute valuations was
given in~\cite{Ben-Zwi2017}.
When the valuations are not restricted to be gross substitute (or more accurate, when Walrasian equilibrium is not guaranteed) the problem that was studied was the maximum social welfare one, rather than the minimal envy free we address here.
For general valuation classes, Blumrosen and Nisan showed
that regardless of time complexity, no ascending auction can
even approximate, to $\Omega(1/\sqrt{m})$, the maximum social welfare,
for $m$ being the number of items~\cite{BlumrosenNi2010}.

Under submodular
valuations, NP-hardness of the maximum social welfare was
proved by Lehmann et. al~\cite{LehmannLeNi2006}; approximation auctions were
given by~\cite{DobzinskiNiSc2005,DobzinskiSc2006,Feige2006,FeigeVo2010,Vondrak2008}; and inapproximation results were presented by~\cite{KhotLiMaMe2008,Vondrak2008,Vondrak2009,ChakrabartyGo2008,DobzinskiVo2013}. 

In this work we ask whether we can come up with an ascending auction that calculates minimal envy free price vector when all valuations are submodular. Our result show that we cannot hope to find such an auction to all submodular valuations.

\vspace{2mm}

\noindent
{\bf Main Results: } No efficient ascending auction can guarantee to find even minimal envy-free prices if all valuations are submodular, assuming $(NPC)\cap (co-NP)=\emptyset$.

\vspace{3.5mm}

\section{Preliminaries and Structural Results}\label{sec:prelim}
In the following we denote by $\bbbr_+$ the set of non-negative reals.
A \emph{combinatorial auction} is a setting where a finite set of
items, $\Omega$, is to be
allocated between a set of players. In what follows it is always
assumed that $|\Omega|=m$, and that there are $n$ players 
associated with valuations $v_i, ~ i=1, \ldots , n$, where $v_i :
2^{\Omega}\mapsto \bbbr_+$. Namely,  for every set $S \subseteq \Omega$,  $v_i(S)$
measures how much the $i-$th player favors the set of items $S$.
Valuation are assumed here to be monotone with respect to inclusion
(aka \emph{free disposal}), and that $v_i(\emptyset)=0$ for every
$i$. The auctioneer sets a price for each item. This is given by a
function $p: \Omega \mapsto \bbbr_+$.\footnote{We will refer to $p$ also as a price-vector.}
We define prices as \emph{real numbers}
and we will equate them later with linear programming variables,
but, as prices should reflect payments, we will assume integrality
when it will be more convenient later on; our results hold also for \emph{real} prices.
Having a price on single items, the price naturally
extends additively to subsets, namely, for $S \subseteq \Omega,~ p(S)
= \sum_{j \in S} p(j)$.
In this work we consider only bidders with submodular valuations.
\begin{definition} [submodular valuation]
A valuation $v:2^{\Omega}\mapsto \bbbr_+$ is a \emph{submodular valuation} if $\forall S,T: v(S\cup T)+v(S\cap T)\leq v(S)+v(T)$.
\end{definition}

Given a set of valuations $\{v_i\}_1^n$ and a price vector $p$, the
\emph{utility} of player $i$ from a set $S \subseteq \Omega$,  denoted
by $u_{i,p}(S)$ is defined by $u_{i,p}(S) = v_i(S) - p(S)$. The
\emph{demand} of a player $i$ and price vector $p$ is a collection of
subsets of items of maximum utility. Namely, $D_i(p) = \{S|\ \forall S'\subseteq \Omega, u_{i,p}(S')\leq u_{i,p}(S)\}$. 

Occasionally we omit subscripts when there is no risk of confusion,
e.g., we use $u_i(S)$ instead of $u_{i,p}(S)$ when $p$ is fixed or
well defined. When the player ($i$) is not important we use $v(S)$ instead of $v_i(S)$,
$u_p(S)$ for $u_{i,p}(S)$ and $D(p)$ for $D_i(p)$.
We use $u_p$, $u_i$ or $u_{i,p}$ to denote the utility of
player $i$, which is its utility from a demand set. So for instance $u_i = u_i(S)$ for
some $S\in D_i(p)$ when the price vector is known to be $p$.

We look at the space of functions on $\Omega$ as ordered by the
domination order, namely, for ${p,q}: \Omega \mapsto \bbbr_+$, $p \leq q$
if for every $j \in \Omega$, $p(j) \leq q(j)$. If $p \leq q$ we
say that $p$ is dominated by $q$.

An \emph{allocation} of the items is a map $A: \{1, \ldots n\}
\mapsto 2^{\Omega}$, which we also denote by $(S_1,\ldots,
S_n)$. Namely, where $A(i)=S_i \subseteq \Omega$. The requirement is
that the sets $\{S_i\}_{1}^{n}$ are pairwise disjoint. 
We think of an allocation
as allocating each player $i$, the  items of the set $S_i$.
 $S_i$ may be empty, for some $i \in [n] = \{1,2,\ldots,n\}$.

\begin{definition} [Envy Free, Minimal Envy Free]
Given a price vector $p$, an allocation is \emph{envy free} w.r.t. $p$
if every player is allocated a demand set. A price vector is envy free if there exists an envy free allocation for it. If also for all $q\leq p$ it holds that $q$ is not an envy free price vector then $p$ is a \emph{minimal} envy free price vector.
\end{definition}

As the name suggests, an envy-free allocation w.r.t. a price vector $p$
signifies an economical situation of `equilibrium' in the sense that
under the corresponding  price $p$, each player is satisfied with its
allocated set.   Clearly, for every set of
valuations, there is an envy-free price vector and an envy-free
allocation for it - simply take $p$ large enough so that the demand
set of every player contains just the empty set. Hence, also a minimal envy free always exists, however, it is not clear if it can be computed and indeed we discuss here the computation (or inability of computing) of a minimal envy free price.

A natural and intuitive search dynamics for any equilibrium is starting with any price vector, then, if there exists an `excess demand' set, increase the prices of its items, if there exists an `excess supply' set, decrease the prices of its items and if there is no such then the price is equilibrium. This process was termed a \emph{t\^{a}tonnement auction}
by Walras~\cite{Walras1874}; however, the notions of excess demand and
excess supply are not easily definable in general, as each bidder may prefer
many different sets for a given price. In fact, our result shows that one cannot define an `excess demand' set when all valuations are submodular.

An ``ascending-auction'' is just the same t\^{a}tonnement auction when it is restricted to only one direction; formally:   In an ascending-auction the auctioneer starts from the
price vector $p_0=\bar{0}$ and at each step $t$, finds whether $p_t$ is
envy free, or increases the prices of some items to obtain the next
price vector $p_{t+1}$. Thus such a process makes sense also as real
economical process.

The notion of ascending-auctions is not well defined. One can
potentially start with $p_0 =\bar{0}$, compute an equilibrium price vector $p^*$ and just set $p_1 = p^*$. A natural ascending-auction should be such that the next step
can be decided from the current step with a very limited knowledge of
the individual valuations (e.g., at step $t$ only access to $D_i(p_t),
i=1, \ldots n$ should be used), and the increase in prices should be
`natural' namely, can be shown to be required, in order to arrive at an
envy-free allocation.  This raises the following discussion.

\section{Witnesses for the non-existence of an allocation}\label{sec:further}

Finding an allocation when all valuations are unit demand is equivalent to finding a matching in a bipartite graph, hence, by Hall's marriage theorem there exists a \emph{witness} for \emph{not} having an allocation~\cite{Hall1935}. This witness is crucial in building an ascending auction as it indicates that the current price vector is not envy free. For all other valuation classes that we know already of an ascending auction the same phenomena occurs, that is, there exists (sometimes implicit) a witness for any non envy free price vector. Therefore, assuming that a current non envy-free price $p$ satisfies $p \leq p^*$ for an envy-free price vector $p^*$, the witness directs the auction to what prices should be increased.

So the problem of whether or not a witness exists, given a price vector, is equivalent to the problem of whether this price vector is \emph{not} an envy free price vector or it is. If we are also guaranteed that the witness is of polynomial size, then the problem of finding a witness is an $NP$ problem, and the problem of finding a minimal envy free price vector is a $co-NP$ problem.
Since an envy free allocation is a (polynomial size) solution for the envy free price vector problem and being minimal is promised by the previous vector being non envy free, the problem of finding a minimal envy free price vector is also in $NP$.
Hence the problem of finding a minimal envy free price vector is in $NP\cap (co-NP)$ and therefore, by a well established assumption cannot be in $NPC$ (see for example~\cite{Sipser1997}).

If we consider the problem of finding an allocation for an equilibrium price vector we can now require that this problem be solvable in polynomial time (or at least not be $NPC$) as a necessary condition for the existence of an ascending auction. Roughgarden and Talgam-Cohen offer looking at two problems in order to know if (Walrasian) equilibrium exists, these are the allocation (aka welfare) problem and the demand (or demand oracle) problem~\cite{RoughgardenTa2015}. They proved that if for a set of valuations, the demand problem is computationally harder than the allocation problem then for this set of valuation an equilibrium is guaranteed. 

Roughgarden and Talgam-Cohen's result suggest a computational method that solves a combinatoric, or even an economic question, namely, the existence of an equilibrium. We offer a computational method to solve an algorithmic problem - the existence of an ascending auction, but note that we are not discussing any computational time here but rather the computational \emph{size}; if the allocation problem is in $NPC$ then there cannot be even a polynomial \emph{size} witness, and therefore not a polynomial size ascending auction.

For general valuations we suggested in~\cite{Ben-ZwiLaNe2013} that there cannot be a polynomial size characterization, since, for instance, \emph{3-dim-matching} is an $NPC$ problem, hence the valuations can be such that the demand sets are restricted to size two (and moreover, each demand set contains at most two pairs), see~\cite{Karp1972}.

We want to show the same result also for submodular valuations. The valuation class we need to suggest then is a submodular one for which the allocation problem is in $NPC$. 
The \emph{Budget Additive} valuation class is fulfill the two requirements and so is the \emph{multi-peak submodular} valuation. For both these valuation classes, by Roughgarden and Talgam-Cohen, a Walrasian equilibrium is not guaranteed. However, for the former, the demand problem is $NPC$ and for the later the demand problem is polynomial. We state that for both we cannot even find a minimal envy free price vector.
\begin{theorem}\label{lower_bound_thm}
There is no polynomial size characterization for a price vector being non envy free
when all valuations are submodular, assuming $(NPC)\cap (co-NP)=\emptyset$.
\end{theorem}
\begin{corollary}
There is no ascending auction that finds even minimal envy free price for submodular valuations.
\end{corollary}
A valuation $v$ is Budget additive if we have a value $v_j$ for each item $j$ and a budget $b$ and for each set $S\subseteq \Omega$ we have its value be
$v(S) = min\{b,\sum_{j\in S}{v_j}\}$.

The multi-peak submodular valuation is defined by attaching one `global' valuation function
with some $k$ other function, given a set system of $k$ sets and a `partition'
of the hypercube into regions which are either `close' to one set or `far' from all.
\begin{definition}[$\epsilon$-close]
A set $S$ is defined to be \emph{$\epsilon$-close} to a set $T$ if $$|S\cap T| - |S\setminus T| > \epsilon |T|.$$
Note that this notion is asymmetric.
\end{definition}
Now given a rational $\epsilon$, a set system ${\cal A} = \{A_1,A_2,...,A_k\}$ such that $\forall i, |A_i| = s$
and $|A_i\cap A_{i'}| \leq \epsilon s$ for $i\neq i'$, we define the multi-peak submodular functions valuation
$v(S)$ in the following manner.
If $\exists A\in {\cal A}$ such that $S$ is $\epsilon$-close to $A$ (unique by~\cite{DobzinskiVo2013}) then:
$$v(S) = \frac{|S\cap A|(2 - \epsilon) + |S \setminus A|(2 + \epsilon)}{2s} + \frac{|S\cap A||S\setminus A|}{s^2} + \frac{\epsilon^2}{4}$$
Else ($S$ is far from any $A\in {\cal A}$) we set: $$v(S) = \frac{|S|}{s} - \frac{|S|^2}{4s^2}$$
This valuation is monotone submodular for every such ${\cal A}, \epsilon, s, k$~\cite{DobzinskiVo2013}.

\begin{proof} (of Theorem~\ref{lower_bound_thm})
Lehmann et. al, prove that the allocation problem of budget additive valuations is $NP-hard$~\cite{LehmannLeNi2006}.
Dobzinski and Vondrak  prove that \emph{multi-peak submodular functions} valuations are submodular and show,
implicitly, that it is $NP-hard$ to decide whether the $0$-price allows envy-free allocation~\cite{DobzinskiVo2013}.

We stress that an ascending auction must decide whether or not the $0$-price is envy-free, hence, assuming $(NPC)\cap (co-NP)=\emptyset$, there cannot be any polynomial size characterization for an ascending auction.
\end{proof}

Roughgarden and Talgam-Cohen also stress that an allocation problem which is not known to be in $P$ will probably not be in $co-NP$~\cite{RoughgardenTa2015}.

Note also that for some valuations we have a polynomial representation, e.g., the valuations
that will result in a \emph{3-dim matching} or the \emph{budget additive} valuations, while other valuation classes have no
short representation, like gross substitute, for instance~\cite{Knuth1974,BalcanHa2018}.
However, unlike budget additive, for gross substitute and for \emph{multi-peak submodular} valuations we have a polynomial time \emph{demand oracle}. A demand oracle for GS is known, see for example Leme~\cite{Leme2017}; we present a polynomial demand oracle for multi-peak submodular valuations:
\begin{theorem}
There exists a polynomial time algorithm that computes a demand set for multi-peak submodular valuations.
\end{theorem}
\begin{proof}
We will run $k+1$ different algorithms to extract the maximum utility
of any one of the original `glued' functions.
The algorithms exploit the fact that the `identity' of an item does not
reflect in the valuation function, hence the price of an item is the sole
criteria for consideration. The different algorithms then only consider
the marginal value added by an item (or a set of two items) against the price
of this item (these two items).
We mark the items by $j_1,j_2,...,j_m$ such that $p_{j_1}\leq p_{j_2}\leq, ... ,\leq p_{j_m}$.
\begin{itemize}
	\item	Algorithm 0:
	\item	find an $l$ such that
	\item     $p_{j_l} \leq 1/s - \frac{2l-1}{4s^2}$
	\item     $p_{j_{l+1}} > 1/s - \frac{2l+1}{4s^2}$
	\item     return $S^0 = \{j_1,j_2,...,j_l\}$
\end{itemize}
For each set $A_i \in {\cal A}$ we mark $A_i$'s items by $j^i_1,j^i_2,...,j^i_s$ such that
$p_{j^i_1}\leq p_{j^i_2}\leq, ... ,\leq p_{j^i_s}$, $\forall r, j^i_r\in A_i$.
We mark the other items by $j^{-i}_1,j^{-i}_2,...,j^{-i}_{m-s}$ such that
$p_{j^{-i}_1}\leq p_{j^{-i}_2}\leq, ... ,\leq p_{j^{-i}_{m-s}}$, $\forall r, j^{-i}_r\notin A_i$.
\begin{itemize}
	\item	Algorithm i:
	\item	for each $l > \epsilon s$ find the largest $\kappa = \kappa(l)$ such that
	\item     $\kappa + \epsilon s < l$
	\item     $p_{j^i_l} + p_{j_\kappa^{-i}} \leq \frac{3}{2s} - \frac{2(\kappa + l) - \epsilon s}{4s^2}$
	\item	mark $S_{l,\kappa}^i = \{j^i_1,j^i_2,...,j^i_l\}\cup\{j^{-i}_1,j^{-i}_2,...,j^{-i}_\kappa\}$
	\item	return $S^i = \arg\max_{S_{l,\kappa}^i\textrm{ is marked}}{\{u_p(S_{l,\kappa}^i)\}}$
\end{itemize}

\begin{proposition}
$\forall S, u_p(S) \leq \max_{i=\{0\}\cup[k]}{\{u_p(S^i)\}}$
\end{proposition}
\begin{proof}
Let $D$ be a demanded set for price vector $p$.
Assume first that $\not\exists A_i\in {\cal A}$ such that
$D$ is close to $A_i$.
We will see that $u_p(D)\leq u_p(S^0)$.
Indeed by definition if $|S^0| < |D|$ then $\exists j\in D\setminus S^0$
such that $p_j \geq p_{j_{l+1}}$ for $l$ as given by Algorithm~$0$ and therefore,
$u_p(D) < u_p(D\setminus \{j\})$. If $|S^0| > |D|$ then $\exists j\in S^0\setminus D$
such that $p_j \leq p_{j_l}$, for the same $l$, hence, 
$u_p(D) < u_p(D\cup \{j\})$. Finally, if $|S^0| = |D|$ then $v(S^0) \geq v(D)$
and $p(S^0) \leq p(D)$.

Assume next that $D$ is close to $A_i \in {\cal A}$. Let
$|D\cap A_i| = l^*$ and $|D \setminus A_i| = \kappa^*$.
Now by assumption $\kappa^* + \epsilon s < l^*$.
If $p_{j^i_{l^*}} + p_{j_{\kappa^*}^{-i}} > \frac{3}{2s} - \frac{2(\kappa + l) - \epsilon s}{4s^2}$
then surely $\exists j_1\in D\cap A_i$, $j_2\in D\setminus A_i$ such that
$p_{j_1}\geq p_{j^i_{l^*}}$ and $p_{j_2} \geq p_{j_{\kappa^*}^{-i}}$ hence,
$u_p(D) < u_p(D\setminus \{j_1,j_2\})$.
Therefore, $p_{j^i_{l^*}} + p_{j_{\kappa^*}^{-i}} \leq \frac{3}{2s} - \frac{2(\kappa + l) - \epsilon s}{4s^2}$,
but then since $\kappa^* + \epsilon s < l^*$ the algorithm either considered $S_{l^*,\kappa^*}^i$
or found a better $\kappa$ for $l^*$. Either way we get that $u_p(D) \leq u_p(S_{l^*,\kappa^*}^i) \leq u_p(S^i)$.
\end{proof}
A polynomial algorithm to compute a demand set then is just
find $S^i$ for $i\in \{0\}\cup [k]$ and return $\arg\max_{i\in\{0\}\cup[k]}{\{u_p(S^i)\}}$.
It is easy to see this is polynomial in $k, s$
\end{proof}


\section*{Appendix - Basic Complexity definitions}\label{complex}
\appendix 
\vspace{0.7cm}
\emph{Complexity theory} tries to define \emph{complexity classes} for different problems, based on the time
needed in order to solve them; by the worst case scenario with respect to the input's size.
A problem is in $P$ if there exists a polynomial time algorithm which solves it.
A problem is in $NP$ if there exists a polynomial time algorithm which given a solution, verifies it.
A problem is in $co-NP$ if there exists a polynomial time algorithm which given a witness for the instance being wrong, verifies it.
It is straightforward that both $P\subseteq NP$ and $P\subseteq co-NP$, and it is believed that $NP\neq (co-NP)$.

A problem is $NP$-Complete ($NPC$) if there is a polynomial time reduction from every other $NP$ problem to it.
If $NP\neq (co-NP)$ then $(NPC)\cap (co-NP)=\emptyset$. For more detailed background see, for example~\cite{Sipser1997}.

\bibliographystyle{plain}
\bibliography{biblio}

\end{document}